%% file: measures_noapp.tex
\documentclass[submission,copyright,creativecommons]{eptcs}
\providecommand{\event}{GandALF 2018} 

\usepackage[utf8]{inputenc}
\usepackage[english]{babel}

\usepackage{amsfonts}
\usepackage{amsmath, amssymb}
\usepackage{amsthm}
\usepackage{graphicx}
\usepackage{wrapfig}

\usepackage{hyperref}
\hypersetup{hidelinks}
\usepackage{underscore}           

\usepackage{enumitem}

\title{On Computing the Measures of~First\=/Order Definable Sets of~Trees}
\author{Marcin Przybyłko\footnote{The author has been supported by Poland’s National Science Centre
        grant no. 2016/21/D/ST6/00491.}
\institute{University of New Caledonia}
\institute{University of Warsaw}
\email{M.Przybylko@mimuw.edu.pl}
}
\def\titlerunning{Measures of First\=/Order Sets of Trees}
\def\authorrunning{M. Przybyłko}

\newtheorem{thm}{Theorem}[section]
\newtheorem{lemma}[thm]{Lemma}

\newtheorem{problem}[thm]{Problem}
\newtheorem{proposition}[thm]{Proposition}
\newtheorem{corollary}[thm]{Corollary}

\newtheorem{fact}[thm]{Fact}

\newtheorem{example}[thm]{Example}

\include{macros}

\begin{document}

\maketitle

\begin{abstract}
We consider the problem of computing the measure of a regular language of infinite binary trees.
While the general case remains unsolved, we show that the measure of a language defined by a~first-order
formula with no descendant relation or by a Boolean combination of conjunctive queries (with descendant relation)
is rational and computable. Additionally, we provide an example of a first-order formula that uses descendant
relation and defines a language of infinite trees having an irrational measure.
\end{abstract}

\section{Introduction}
The problem of computing a~measure of a~set can be seen as one of the fundamental problems considered in the study of probabilistic systems.
This problem has been studied mostly implicitly, as it is often one of the intermediary steps in solving stochastic games, cf.~\cite{chatterjeeStochasticRegular},
answering queries in probabilistic databases, cf.~\cite{suciuProbDatabases}, or model checking for stochastic branching processes, cf.~\cite{chenStochasticBranchingProcesses}.

To us, this problem naturally arises in the study of stochastic games.
Many of the games considered in literature can be seen as instances of the stochastic version of Gale-Stewart games~\cite{galeGames}.
Such games have winning conditions expressed as sets of winning plays, i.e. a set of (in)finite words representing winning plays.
Hence, computing the value of a~stochastic game involves computing the measure of the winning set with respect to the probabilistic space generated
by the stochastic elements of the game. 

Mio~\cite{mioBranchingGames} introduced branching games, i.e. stochastic games for which
the plays are represented as (in)finite trees, rather than words.
Then, the questions concerning whether a~set of trees has a~measure~\cite{michalewskiMeasure},
and whether that measure can be computed~\cite{michalewskiCompMeasure} have been raised and partially answered.

In this work, we focus our attention on the problem of computing the uniform measure of a set of labelled infinite trees defined by 
a~first-order formula using child and descendant relations.

\paragraph{Related work}
The problem of computing the measure of an arbitrary regular language of trees has been already, explicitly or implicitly, studied.
Gogacz et al.~\cite{michalewskiMeasure} prove that regular languages of trees are universally measurable.
In the case of infinite trees,
Chen et al.~\cite{chenStochasticBranchingProcesses} show that the measure of a language accepted by deterministic automaton is computable;
Michalewski and Mio~\cite{michalewskiCompMeasure} extend the class of languages with computable measure to the class of languages defined by game automata.
In the case of finite trees, Amarilli et al.~\cite{amarilliProvenance} show that with measures defined by fragments of the probabilistic XML, i.e. where the support of the measure consists of the trees of bounded depth, the measure is computable for arbitrary regular sets of trees. 
In the case of regular languages of infinite words, Staiger~\cite{staigerMesureOnWords} shows that the measure of an arbitrary 
regular language of words is computable.
Note that, in all the above results the inherent deterministic nature of the involved automata plays an important role.

The problem of computing the measure of tree a~language has been, implicitly, considered in probability games.
The problem is a special case of computing the value of a probability game when the strategies of players are already chosen.
In the case of infinite trees, Przybyłko and Skrzypczak~\cite{przybylkoRegularBranchingGamesMFCS} consider branching games with regular wining sets.
In the case of words, for the survey of probabilistic $\omega$-regular games on graphs see e.g.~Chatterjee and Henzinger~\cite{chatterjeeStochasticRegular}.

The problem of computing the measure can be also seen as the problem of query evaluation on probabilistic databases, e.g.
Amarilli et al.~\cite{amarilliCQsOnProbabilisticGraphs} enquire into the evaluation problem of conjunctive queries over probabilistic graphs.
For an introduction to probabilistic databases see e.g.~\cite{suciuProbDatabases}.

\paragraph{Our contribution}
We provide algorithms to compute the measures of tree languages definable by some restricted classes of first-order formulae.
We show that, in the case of first-order formulae using unary predicates and child relation, the standard measure can be computed in three-fold exponential time.
Moreover, in the case of Boolean combinations of conjunctive queries using unary predicates, child relation, and descendant relation, the measure can be computed in exponential space.

We also provide an example of a first-order formula over a two letter alphabet
for which the defined language has an irrational standard measure.
An example of a regular language with an irrational standard measure was already presented in~\cite{michalewskiCompMeasure}, however that language is not first\=/order definable.

\paragraph{Organization of the paper}
In Section~\ref{sec:prelim} we define basic notions used in this article.
In Section~\ref{sec:examples} we showcase the properties of the standard measure on 
selected examples. The computability of the measure of the regular languages defined by
first-order formulae is discussed in Section~\ref{sec:FOMeasures}.
The computability of the measure of the regular languages defined by
conjunctive queries is discussed in Section~\ref{sec:CQsMeasure}.
In the last section, we discuss obtained results and propose some directions of future research.

\section{Preliminaries}
\label{sec:prelim}

In this section we present crucial definitions used throughout this work.
We assume basic knowledge of logic, automata, and measure.
For introduction to logic and automata see~\cite{thomasLanguages}, for introduction
to topology and measure see~\cite{kechrisDescriptive}.

\paragraph{Words and trees}
An alphabet is any non-empty finite set.
A~\emph{word} is a~partial function \parfun{w}{\mathbb{N}}{\Gamma}, such that $\mathbb{N}$ is the set of natural numbers, $\Gamma$ is an~alphabet, and
the domain $\dom{w}$ of $w$ is $\leq$-closed. 
By $|w|$ we denote the length of the word $w$, i.e. the size of its domain.
By $\varepsilon$ we denote the empty word, i.e. the unique word of length $0$.
If the domain of a~word $w$ is finite, then the word is called finite; otherwise, it is called infinite.
The set of all finite words over an~alphabet $\Gamma$ is denoted $\Gamma^{\ast}$,
the set of all infinite words over an~alphabet $\Gamma$ is denoted $\Gamma^{\omega}$.
Let $n \in \mathbb{N}$ and $\bowtie \in \{<, \leq, =, \geq, > \}$, then the set of all words over an~alphabet $\Gamma$ of length $l$ such that $l \bowtie n$ is denoted $\Gamma^{\bowtie n}$.

A~word $w$ is called a~prefix of a~word $v$, denoted $w \sqsubseteq v$, if $\dom{w} \subseteq \dom{v}$ and for every $i \in \dom{w}$ we have that
$w(i) = v(i)$. By $w \cdot v$, or simply $wv$,  we denote the concatenation of the words $w$ and $v$.

A tree is any~partial function \parfun{t}{\setPositions}{\Gamma}, where the domain $\dom{t}$ is prefix-closed and $\Gamma$ is a~finite alphabet.
The elements of the set $\setD$ are called directions (left and right, respectively) and the elements of the set $\setPositions$ are called positions.
For a given tree $t$, the elements of the set $\dom{t}$ are called nodes of the tree $t$, or nodes for short.
A tree $t$ is either finite, if its domain $\dom{t}$ is finite, or infinite.
The tree $t$ is called a full binary tree of height $k$ if $\dom{t} = \setD^{\le k}$.
A~tree $t$ is called a full binary tree if~$\dom{t} = \setPositions$.
Let $\Gamma$ be an alphabet. The set of all trees over an alphabet $\Gamma$ is denoted by $\trees{\Gamma}$;
the set of all finite trees by $\trees{\Gamma}^{F}$;
the set of all full binary trees of height $k$ by $\trees{\Gamma}^{=k}$;
the set of all full binary trees by $\trees{\Gamma}^{\omega}$.
A tree $t_1$ is called a~prefix of a~tree $t_2$, denoted $t_1 \weakAncestor t_2$, if $\dom{t_1} \subseteq \dom{t_2}$ and for every $u \in \dom{t_1}$ we have that
$t_1(u) = t_2(u)$. 
We say that a tree $t_2$ is a sub-tree of $t_1$ in node $u \in \dom{t_1}$ if $u \cdot \dom{t_2} \subseteq \dom{t_1}$ and for every $v \in \dom{t_2}$ we have that $t_1(uv) \subseteq t_2(v)$.
Let $t_1$ be~a~tree, by~$\subtree{t_1}{u}$ we denote the $\weakAncestor$-biggest sub-tree of $t_1$ in the node $u \in \dom{t_1}$, i.e. biggest with the respect to the containment of the domains.
For a~tree $\treeZ$ and a~position $u$, by $\setBallInNode{\treeZ}{u}$  we denote the set of full binary trees in~which $t$ is a sub-tree in node $u$, i.e. the set
\begin{equation}
\label{eq:ball}
\setBallInNode{\treeZ}{u} \eqdef \{\treeZ' \in \treesInfinite{\nhasBlank{\Gamma}} \mid \isPrefix{\treeZ}{\subtreeAtNode{\treeZ'}{u}}\},
\end{equation} with $\mathbb{B}_{\treeZ} \eqdef \setBallInNode{\treeZ}{\varepsilon}$.

\paragraph{Logic}
A tree $t$ over an alphabet $\Gamma$ can be seen as a~relational structure
$t = \langle {\dom{t}}, {s_{\dL}}, {s_{\dR}}, {s}, {\ancestor}, {(a^{t})_{a \in \Gamma}} \rangle$, where

\begin{itemize}
    \item $\dom{t}$ is the domain of $t$;
    \item $s_{\dL},s_{\dR}  \subseteq \dom{t} \times \dom{t}$ are the left child relation ($u\ {s_\dL}\ u\cdot\dL$) and right child relation ($u\ {s_\dR}\ u\cdot\dR$), respectively;
    \item $s$ is the child relation $s_{\dL} \cup s_{\dR}$;
    \item $\ancestor$ is the ancestor relation, i.e. the transitive closure of the relation $s$;
    \item $a^{t} \subseteq \dom{t}$ is a~subset of $\dom{t}$, for ${a \in \Gamma}$, and the family of sets  $(a^{t})_{a \in \Gamma}$ is a partition of $\dom{t}$.
\end{itemize}
The partition $(a^{t})_{a \in \Gamma}$ induces a labelling function $\fun{\lambda_t}{\dom{t}}{\Gamma}$ in the natural way, i.e. $\lambda_{t}(u) = a$ if and only if $u \in a^{t}$.

\paragraph{Regular languages}
Formulae of Monadic Second-Order logic (MSO) can quantify over positions in trees $\exists x, \forall x$ and over sets of positions $\exists X, \forall X$.
A First-Order (FO) formula is an~MSO formula that does not quantify over the sets of positions.
A sentence is a formula with no free variables.

We say that an~MSO formula $\varphi$ is over a~signature $\Sigma$ if $\varphi$ is a well\=/formed formula
built from the symbols in $\Sigma$ together with the quantifiers and logical connectives.
Let $\Gamma$ be an alphabet, in this paper, we consider only the formulae over the signatures $\Sigma$
such that $\Sigma \subseteq \{ s_{\dL}, s_{\dR}, s, \ancestor, \rootP \} \cup \Gamma$.

Let $\varphi$ be a~first-order formula.
We write $t,v \models \varphi(x_1, \dots, x_k)$, if the tree $t$, as a logical structure, with the~valuation $v \in \dom{t}^k$ satisfies the formula $\varphi(x_1,\dots, x_k)$.
If $\varphi$ is a sentence, we simply write $t \models \varphi$.
We say that a~formula $\varphi(x_1, \dots, x_k)$ is satisfiable if there is a tree $t$ and a~tuple $v \in \dom{t}^k$ such that $t,v \models \varphi(x)$.

Let $\Gamma$ be an alphabet, the set defined by an~MSO sentence $\varphi$, denoted $\tL(\varphi)$,
 is the set of all full binary trees over the alphabet $\Gamma$ that satisfy $\varphi$, i.e.
 $\tL(\varphi) \eqdef \{ t \in \trees{\Gamma}^{\omega} \mid t \models \varphi\}$.
A~language defined by an~MSO formula is called regular. This definition of regular languages of trees is equivalent to the automata based definition, cf. e.g.~\cite{thomasLanguages}. 

\paragraph{Measure}
The set of all full binary trees over an alphabet $\Gamma$, denoted $\TT^{\omega}_{\Gamma}$, is the set of all functions \fun{t}{\setPositions}{\Gamma}.
This set can naturally be enhanced with a topology in such a way that it becomes a homeomorphic copy of the Cantor set, see~Gogacz et al.\cite{michalewskiMeasure} for 
more detailed definitions.
The \emph{uniform} (or \emph{standard}) measure $\mu^{\ast}$ defined on the set of trees $\TT^{\omega}_{\Gamma}$
is the probability measure such that for every finite tree $\treeZ \in \treesInfinite{\Gamma}$ we have that
$\standardMeasure{\bB_{\treeZ}} = |\Gamma|^{-|\nodes{\treeZ}|}$. In other words, this measure is such that for every node $u \in \dS^{\ast}$ and label $a \in \Gamma$ 
the probability that in a random full binary tree $\treeZ$ the node $u$ is labelled with the letter $a$ is $\frac{1}{|\Gamma|}$, i.e. $\standardMeasure{\{ t \in \treesInfinite{\Gamma} \mid t(u)=a \}} = \frac{1}{|\Gamma|}$.

By the following theorem we conclude that every regular language of trees $L$ has well defined standard measure $\standardMeasure{L}$.
\begin{thm}[\cite{michalewskiMeasure}]
    Every regular language $L$ of infinite trees is universally measurable, i.e.~for every Borel measure $\mu$ on the set of trees, we know that $L$ is $\mu$\=/measurable.
\end{thm}

Hence, the following problem is well-defined.

\begin{problem}[The $\standardMeasure{\text{MSO}}$ problem]
    \label{problem:computeMeasure}
    Is there an algorithm that given an~MSO formula $\varphi$ computes $\standardMeasure{\varphi}$?
\end{problem}

With the $\standardMeasure{\text{MSO}}$ problem we associate the following decision problem.

\begin{problem}[The positive $\standardMeasure{\text{MSO}}$ problem]
    \label{problem:positiveMeasure}
    Given an~MSO formula $\varphi$, decide whether $\standardMeasure{\varphi} > 0$.
\end{problem}

If $\CC$ is a class of \text{MSO} formulae, then by \emph{the (positive) $\standardMeasure{\CC}$ problem}, 
we understand the above where possible input formulae are restricted to the class $\CC$.
If we restrict the class $\CC$ to formulae over the signature $\Sigma$ we denote it by $\CC(\Sigma)$.

The problem, in this form, was stated by Michalewski and Mio~\cite{michalewskiCompMeasure}.
It is open in the general case, but some partial results have been obtained, see paragraph \emph{Related work} for details.

\section{Measures of simple languages}
\label{sec:examples}

To better understand the properties of the standard measure, we start with some examples of sets of infinite trees and their measures.
The examples will be used in the proofs in the following sections.

\begin{lemma}
    \label{lemma:basicMeasuresLemma}
    Let $\treeZ$ be a binary tree over the alphabet $\Gamma$, $u \in \setNodes$ a~position.
    \begin{enumerate}\label{enum:basiecMeasuresLemma}
        \item \label{item:finitePrefix} If $\treeZ$ is finite and $L_{\text{fp}} = \setBallInNode{\treeZ}{u} = \{\treeZ' \in \treesInfinite{\nhasBlank{\Gamma}} \mid \isPrefix{\treeZ}{\subtreeAtNode{\treeZ'}{u}}\}$, then $\standardMeasure{L_{\text{fp}}} = {\Gamma}^{-|\nodes{\treeZ}|}$.
        \item \label{item:finiteSubtree} If $\treeZ$ is finite and 
        $L_{\text{fs}} \eqdef \{\treeZ' \in \treesInfinite{\Gamma} \mid \exists v. (u \weakAncestor v) \land (\treeZ \weakAncestor \subtreeAtNode{\treeZ'}{v})\}$, then $\standardMeasure{L_{\text{fs}}} = 1$.
        \item \label{item:infinitePrefix} If $\treeZ$ is infinite and $L_{\text{ip}} = \setBallInNode{\treeZ}{u} = \{\treeZ' \in \treesInfinite{\nhasBlank{\Gamma}} \mid \isPrefix{\treeZ}{\subtreeAtNode{\treeZ'}{u}}\}$, then $\standardMeasure{L_{\text{ip}}} = 0$.
    \end{enumerate}
\end{lemma}

\begin{proof}
    The proof of Item~\ref{item:finitePrefix} is straightforward. To prove Item~\ref{item:finiteSubtree}, for $i \geq 0$ let $L_i$ be the language $L_i \eqdef \setBallInNode{\treeZ}{v\dL^i\dR} = \{\treeZ' \in \treesInfinite{\nhasBlank{\Gamma}} \mid \isPrefix{\treeZ}{\subtreeAtNode{\treeZ'}{(v\dL^i\dR)}}\}$. Then, for every $j \geq 0$ we have that $L_j \subseteq L_{\text{fs}}$ and, in consequence, 
    $\overline{L_{\text{fs}}} \subseteq \bigcap_{j \geq i \geq 0} \overline{L_i}$. Hence, for every $j \geq 0$ we have that
    
    \[
    1 - \standardMeasure{L_{\text{fs}}} = \standardMeasure{\overline{L_{\text{fs}}}} \leq \standardMeasure{\bigcap_{j > i \geq 0} \overline{L_i}} \leq \big(1 - |\Gamma|^{|\nodes{\treeZ}|}\big)^{j},
    \]
    where the last inequality follows from the fact that the nodes $v\dL^l\dR$ and $v\dL^k\dR$ are incomparable for $k \not = l$, thus $L_i$ are independent sets and
    \[
    \standardMeasure{\bigcap_{j > i \geq 0} \overline{L_i}} = \prod_{0 \leq i < j} \standardMeasure{\overline{L_i}}=  \prod_{0 \leq i < j} (1 -  |\Gamma|^{|\nodes{\treeZ}|}) = \big(1 - |\Gamma|^{|\nodes{\treeZ}|}\big)^{j}
    \]
    Taking the limit, we conclude Item~\ref{item:finiteSubtree}.
    
    To prove Item~\ref{item:infinitePrefix}, let $t_i$ be a sequence of finite trees such that for every $i \geq 0$ we have that $\treeZ_i \prefix \treeZ_{i+1} \prefix \treeZ$ and $|\nodes{\treeZ_i}| < |\nodes{\treeZ_{i+1}}|$.
    Since the sequence of sets $\setBallInNode{\treeZ_i}{v}$ is decreasing and its limit contains the set $\setBallInNode{\treeZ}{v}$, i.e. $\setBallInNode{\treeZ_i}{v} \supseteq \setBallInNode{\treeZ_{i+1}}{v} \supseteq \setBallInNode{\treeZ}{v}$, we have that 
    $\standardMeasure{\setBallInNode{\treeZ}{v}} \leq \lim\limits_{i \to +\infty}\standardMeasure{\setBallInNode{\treeZ_i}{v}} = \lim\limits_{i \to +\infty} {|\Gamma|}^{-|\nodes{\treeZ_i}|} =  0$.
\end{proof}

\begin{example}
    \label{ex:someMeasuresOfSets}
    Let $\Gamma = \{a,b,c\}$.
    \begin{enumerate}\label{enum:someMeasuresOfSets}
        \item \label{item:laprefixes} If $L_a$ is the language of trees over the alphabet $\Gamma$ with arbitrarily long sequences of $a$-labelled nodes, i.e., 
        $L_a = \{\treeZ \in \treesInfinite{\Gamma}\mid \forall k \ge 0.\ \exists w,v \in \dS^{\ast}. \big(
        \left(|v| \geq k \right) \land \forall\ u \sqsubseteq v.\ \treeZ(wu)  = a\big)\}
        $, then $\standardMeasure{L_a} = 1$.
        \item \label{item:la3path}If $L_{a3}$ is the language of trees over the alphabet $\Gamma$ with an~infinite $\{a\}$-labelled path, i.e. 
        $
        L_{a3} = \{\treeZ \in \treesInfinite{\Gamma}\mid \exists w \in~\dS^{\omega}. \forall u \sqsubseteq w.\ \treeZ(u) = a\}
        $, then $\standardMeasure{L_{a3}}=0$.
        \item \label{item:la2path} If $L_{a2}$ is the language  of trees over the alphabet $\{a,b\}$ with an~infinite $\{a\}$-labelled path, i.e. 
        $
        L_3 = \{\treeZ \in \treesInfinite{\{a,b\}}\mid \exists w \in~\dS^{\omega}. \forall u \sqsubseteq w.\ \treeZ(u) = a\}
        $, then $\standardMeasure{L_{a2}}=0$.
        \item \label{item:labpath} If $L_{ab}$ is the language of trees over the alphabet $\Gamma$ with an~infinite $\{a,b\}$-labelled path, i.e. 
        $
        L_{ab} = \{\treeZ \in \treesInfinite{\Gamma}\mid \exists w~\in~\dS^{\omega}. \forall u \sqsubseteq w.\ \treeZ(u) \in \{a,b\}\}
        $, then $\standardMeasure{L_{ab}}=\frac{1}{2}$.
    \end{enumerate}
\end{example}
\begin{proof}[Calculating the measures]

    To show Item~\ref{item:laprefixes}, let $\treeZ^i$ be a complete tree of height $i$ with every node in $\nodes{\treeZ^i}$ labelled $a$
    and let $L^i$ be the language of trees having $t^i$ as a~sub-tree.
    Then by Lemma~\ref{lemma:basicMeasuresLemma} part $\ref{item:finiteSubtree}$ we have that $\standardMeasure{L^i} = 1$.
    Moreover, $\bigcap_{i \geq 1} L^i \subseteq L_{a}$ and $L^{i+1} \subseteq L^i$.
    Since the measure is monotonically continuous, we have that 
    \begin{equation}
    \standardMeasure{L_{a}} \geq \standardMeasure{\bigcap_{i \geq 1} L^i} = \lim\limits_{n \to +\infty} \standardMeasure{\bigcap_{n \geq i \geq 1} L^i} = 1.
    \end{equation}
    
    Let $\phi(\treeZ)$ stay for  ``there is an~infinite $a$-labelled path in the tree $\treeZ$'' then
    Item~\ref{item:la3path} follows from the fact that the language in question is regular, thus measurable, and its measure satisfies the following equation.
    \begin{align*}
    \standardMeasure{L_{a3}} = \standardMeasure{\treeZ(\varepsilon) {\not=} a} &\cdot 0\; + \\
    \standardMeasure{\treeZ(\varepsilon) {=} a}&\cdot
    \big(
    \standardMeasure{\phi(\subtreeAtNode{\treeZ}{\dL})}
    +
    \standardMeasure{\phi(\subtreeAtNode{\treeZ}{\dR})}
    - \standardMeasure{\phi(\subtreeAtNode{\treeZ}{\dL}) \land \phi(\subtreeAtNode{\treeZ}{\dR})}
    \big)
    \end{align*}
    Thus, we get the equation
    \begin{equation}
    \standardMeasure{L_{a3}} = \frac{1}{3}\cdot \big(2\standardMeasure{L_{a3}} - \standardMeasure{L_{a3}}^2\big)=\frac{2}{3}\standardMeasure{L_{a3}}-\frac{1}{3}\cdot\standardMeasure{L_{a3}}^2
    \end{equation}
    implying that $\standardMeasure{L_{a3}} =  0$, since the measure cannot be negative.
    
    Similarly, in Item~\ref{item:la2path} we get the equation
    \begin{equation}
    \label{example:measure4}
    \standardMeasure{L_{a2}} = \frac{1}{2}\cdot \big(2\standardMeasure{L_{a2}} - \standardMeasure{L_{a2}}^2\big)=\standardMeasure{L_{a2}}-\frac{1}{2}\cdot\standardMeasure{L_{a2}}^2
    \end{equation}
    implying that $\standardMeasure{L_{a2}} =  0$.
    
    In Item~\ref{item:labpath}, we get the equation
    \begin{equation}
    \label{example:measure3}
    \standardMeasure{L_{ab}} = \frac{2}{3}\cdot \big(2\standardMeasure{L_{ab}} - \standardMeasure{L_{ab}}^2\big)=\frac{4}{3}\standardMeasure{L_{ab}}-\frac{2}{3}\cdot\standardMeasure{L_{ab}}^2
    \end{equation}
    implying that either $\standardMeasure{L_{ab}} =  \frac{1}{2}$ or $\standardMeasure{L_{ab}} =  0$.
    Thus we need to look at this example a bit more carefully.
    Consider a~sequence of languages $A^i$, where $A^0 = \trees{\Gamma}$
    and $A^i$ is the language such that there is a $\{a,b\}$-labelled path of length $i$ beginning at the root.
    Then, by K\"{o}nig's lemma we have that $\bigcap_{i \geq 1} A^i = L_{ab}$. Moreover, for every $i > 0$ we have that $A^{i+1} \subseteq A^i$ and 
    \begin{equation}
    \label{example:recursiveMeasure}
    \standardMeasure{A^{i+1}} = \frac{2}{3}\cdot \big(2\standardMeasure{A^i} - \standardMeasure{A^i}^2\big)=\frac{4}{3}\standardMeasure{A^i}-
    \frac{2}{3}\cdot\standardMeasure{A^i}^2.
    \end{equation}
    Now, note that if $\standardMeasure{A^i} \ge \frac{1}{2}$, then $\standardMeasure{A^{i+1}} \ge \frac{1}{2}$. Indeed, the quadratic function $f(x) = \frac{2}{3}(2x - x^2)$ is monotonically  increasing on the interval $[-\infty, 1] $ and we have that $f(1) = \frac{2}{3}$ and $f(\frac{1}{2}) = \frac{1}{2}$. Since $\standardMeasure{A^0} = 1 \ge \frac{1}{2}$, we conclude that $\standardMeasure{L_{ab}} =  \frac{1}{2}$. 
\end{proof}

\section{First-order definable languages and their standard measures}
\label{sec:FOMeasures}

The ideas presented in both Lemma~\ref{lemma:basicMeasuresLemma} and Example~\ref{ex:someMeasuresOfSets} allow us to 
compute the measures of tree languages defined by some $\fo$ formulae.
\begin{thm}
	\label{thm:FOcomputable}
	Let $\varphi$ be a $\fo$ sentence over the signature $\Gamma \cup \{\rootP, s_\dL, s_\dR, s\}$. Then, the measure $\standardMeasure{\tL(\varphi)}$ is rational and computable in three-fold exponential time.
\end{thm}
The proof utilises the \emph{Gaifman locality} to partition the formula into two separate sub-formulae.
Intuitively, one sub-formulae describes the neighbourhood of the root while the other describes the tree ``far away from the root''.

Before we prove the above theorem, we define the idea of a root formula, i.e. a formula that necessarily describes the neighbourhood of the root.
Let $\AA$ be a logical structure. The \emph{Gaifman graph} of $\AA$ is the undirected graph $G^{\AA}$ where the set 
of vertices is the universe of $\AA$ and there is an edge between two vertices in $G^{\AA}$ if there is a relation $R$ in $\AA$ 
and a tuple $x \in R$ that contains $u$ and $v$.
The distance $d(u,v)$ between two elements $u,v$ of the universe of $\AA$ is the distance between $u,v$ in the Gaifman graph.

Before we proceed, let us note that in this section we disallow the use of $\ancestor$ relation in formulae.
Hence, the Gaifman graf of a tree $t$ is induced by the child relations only, and so is the notion distance.

We say that a first-order formula $\varphi(x)$ is a $r$-local formula around $x$ if the quantifiers are restricted to
$r$-neighbourhood of $x$, i.e. if $\varphi(x)$ uses the quantifiers $\forall^{\le r}$ and $\exists^{\le r}$ defined as follows:
$\exists^{\le r} y. \psi(y) \eqdef \exists y . \psi(y) \land d(x,y) \leq r$ and 
$\forall^{\le r} y. \psi(y) \eqdef \forall y . (d(x,y) \leq r) \to \psi(y)$, where $d(x,y) \leq r$ is a first-order formula stating
that the distance between $x$ and $y$ is at most $r$.

We say that a first-order sentence $\varphi$ is a \emph{basic $r$-local sentence} if it is of form
\begin{equation}
\label{eq:basicLocalFormula}
\exists x_1,\dots,x_n \big( \bigwedge_{i=1}^{n} \varphi_i(x_i) \land \bigwedge_{0 \le i<j \le s} d(x_i,x_j) > 2r, \big)
\end{equation}
where $\varphi_i(x)$ are $r$-local formulae around $x$ and $d(x,y) > 2r$ is a first-order formula stating that 
the distance between $x$ and $y$ is strictly greater than $2r$.

Let $\treeZ \in \trees{\Gamma}$ be a tree and let $\varphi$ be a basic $r$-local sentence, then
$\treeZ \models \varphi$ if and only if there is a function $\fun{\tau}{x_1,\dots,x_n}{\dom{t}}$
mapping variables $x_1,\dots, x_n$ to the nodes of $\treeZ$ so that for every $i \in \{1,\dots,n\}$
we have that $\treeZ,\,\tau(x_i) \models \varphi_i(x_i)$.

\begin{thm}[Gaifman]
    \label{thm:gaifman}
    
    Every first-order sentence is equivalent to a Boolean combination of basic $r$-local sentences, where $r$ is a number depending
    on the size of the formula.
    Furthermore, $r$ can be chosen so that $r \leq 7^{\textit{qr}(\varphi)}$, where $\textit{qr}(\varphi)$ is the quantifier rank of $\varphi$.
    \hide{
        Let $\sigma$ be relational. Then every $\fo$ formula $\varphi(x)$ over $\sigma$ is equivalent to a Boolean combination of the following:
        \begin{itemize}\itemsep=0pt
            \item local formulae $\varphi^{r}(x)$ around $x$,
            \item 
        \end{itemize}
        Furthermore
        \begin{itemize}\itemsep=0pt
            \item the transformation to the Boolean combination is effective;
            \item if $\phi$ is a sentence then only sentences occur in the combination;
            \item if the quantifier rank is $k$ and the size of $x$ is $n$ then the bounds on $r$ and $s$ are $r \leq 7^k$, $s\leq k+n$.
        \end{itemize}
    }
\end{thm}

As proved by Heimberg et al., cf.~\cite{heimbergOptimalGaifman}, the translation to Gaifman normal form can be costly.

\begin{thm}[\cite{heimbergOptimalGaifman}]
    \label{thm:gaifmanCost}
There is a three-fold exponential algorithm on structures of degree 3 that transforms a first-order formula into its Gaifman normal form.
Moreover, there are first-order formulae for which the three-fold exponential blow-up is unavoidable.
\end{thm}

Let $\psi(x)$ be a $r$-local formula around $x$.
We say that $\psi(x)$ is a \emph{root} formula if 
for every tree $\treeZ \in \trees{\nhasBlank{\Gamma}}$ and every node $u \in \dS^{\ast}$ 
if $\treeZ, u \models \psi(x)$ then $\td(u,\varepsilon) < r $.
Notice that every unsatisfiable formula is, by the
definition, a~root formula.
Let $\varphi$ be a basic $r$-local sentence. 
We say that $\varphi_i(x)$ for $i \in \{1,\dots,n\}$ \emph{is~a~root formula of~$\varphi$} if
$\varphi_{i}(x)$ is a root formula.

\begin{fact}
	\label{fact:uniqueRootFormula}
	For every satisfiable basic $r$-local sentence there is at most one root formula.
\end{fact}

With the above definitions, we can describe the connection between the basic local sentences
and the standard measure.

\begin{lemma}
	\label{lemma:basicSentenceMeasure}
	Let $\varphi$ be a basic $r$-local sentence, i.e. as in Equation~\eqref{eq:basicLocalFormula}.
	If $\varphi$ is 
	\begin{itemize}\itemsep=0pt
		\item not satisfiable, then $\standardMeasure{\tL(\varphi)} = 0$,
		\item satisfiable and has no root formula, then $\standardMeasure{\tL(\varphi)} = 1$,
		\item { satisfiable and has a root formula $\varphi^{\ast}$, then for every $\treeZ^r$ that is a~complete tree of~height~$2r+1$
			\[
			\standardMeasure{\tL(\varphi) \cap \bB_{\treeZ^r}} =
			\begin{cases}
			\standardMeasure{\bB_{\treeZ^r}} & \quad \text{if there is } u\in\dS^{\leq r} \text{ such that }\quad \treeZ^r{,}u \models  \varphi^{\ast}(x);\\
			0 & \quad \text{otherwise.}
			\end{cases}
			\]
		}
	\end{itemize}
\end{lemma}
\begin{proof}
	If $\varphi$ is not satisfiable then $\tL(\varphi) = \emptyset$ and $\standardMeasure{\tL(\varphi)} = 0$.
	Therefore, let us assume that $\varphi$ is satisfiable.
	By Fact~\ref{fact:uniqueRootFormula} we know that there is at most one root formula in $\varphi$.
	Let $I$ be the set of indices of not root formulae, i.e.
	for $i \in I$ we have that $\varphi_i$ is not a root formula. Since $\varphi$ is satisfiable then
    for every $i \in I$ there is a finite tree $\treeZ_i \in  \trees{\hasBlank{\Gamma}}$ and a node $u_i\in\dS^\ast$ of length $|u_i|>r$, such that $\treeZ_i,u_i \models \varphi_i(x)$, and the set
    $\nodes{\treeZ_i}$ contains the $r$-neighbourhood of $u_i$.
	
	Let $W = \{v_i\}_{i=1}^n$ be a set of $n$ $\sqsubseteq$-incomparable nodes such that for $i \in I$ we have that $|v_i| > 2r$.
	Let $F = \bigcap_{i \in I} L_i$ where $L_i$ is the set of trees having $\treeZ_i$ as a sub-tree rooted below the node $v_i$, i.e. $L_i \eqdef \{\treeZ' \in \trees{\nhasBlank{\Gamma}} \mid \exists u. (v_i \ancestor u) \land (\treeZ_i \ancestor \subtreeAtNode{\treeZ'}{u})\}$.
	Since, by Lemma~\ref{lemma:basicMeasuresLemma}, every $L_i$ has measure $1$, we have that $\standardMeasure{F}=1$.
	Moreover, for every tree $\treeZ \in F$ and index $i\in I$  there is a node $v_i' \in \dS^\ast$ such that $\td(v_i,v'_i)> r$, $v_i \sqsubseteq v_i'$ and $\treeZ,v_i' \models \varphi_i(x)$.
	
	Now, if there is no root formula in $\varphi$, i.e. $I = \{1\dots,n\}$, then $F \subseteq \tL(\varphi)$.
	Indeed, let $\treeZ \in F$ then for $i\not=j$ we have that $\td(v_i',v_j')>2r$ and we can infer that $\treeZ \models \varphi$. Hence, the sequence of inequalities
	\[1 = \standardMeasure{F}  \leq \standardMeasure{\tL(\varphi)} \le 1\] is sound and proves the second bullet.
	
	On the other hand, let there be a~root formula in $\varphi$. Without loss of generality $\varphi_1$ is the root formula and $I = \{2,\dots,n\}$. Let $F$ and $v_i'$s be as before and let  $\treeZ \in F \cap \mathbb{B}_{\treeZ^r}$.	
	If there is $u_1\in\dS^{\leq r}$ such that $\treeZ^r, u_1 \models  \varphi^{\ast}(x)$ then we take $v_1' \eqdef u_1$.
	Now, again, for $i\not=j$ we have that $\td(v_i',v_j')>2r$ and for all $i\in I$ we have that $\treeZ, v_i' \models \varphi_i(x_i)$. 
	In other words, if there is $u_1\in\dS^{\leq r}$ such that $\treeZ^r, u_1 \models  \varphi^{\ast}(x)$ then
	$F \cap \mathbb{B}_{\treeZ^r} \subseteq \tL(\varphi) \cap \mathbb{B}_{\treeZ^r}$.
	Moreover, since  $F$ is of measure $1$, the following sequence of inequalities is sound
	\[\standardMeasure{\mathbb{B}_{\treeZ^r}} = \standardMeasure{F \cap \mathbb{B}_{\treeZ^r}}  \leq \standardMeasure{\tL(\varphi) \cap \mathbb{B}_{\treeZ^r}} \leq \standardMeasure{\mathbb{B}_{\treeZ^r}}.\]
	
	Furthermore, if there is no such $u_1$ then $\varphi_1$ is not satisfiable in $\treeZ^r$.
	Since $\varphi_1$ is a~root formula, we have that $\standardMeasure{\mathbb{B}_{\treeZ^r} \cap \tL(\varphi)} = 0$, which concludes the proof.
\end{proof}

Intuitively, the above lemma states that when we consider the uniform measure and a~basic $r$-local sentence, the behaviour of the sentence is almost surely defined by the neighbourhood of the root. 
This intuition can be formalised as follows.

\begin{lemma}
	\label{lemma:uniqueRootFormula}
	Let $\varphi$ be a basic $r$-local sentence.
	Then, there is a~sentence $\varphi^{\ast}$ such that for every complete tree $\treeZ^r$ of height $2r+1$
	\[ \standardMeasure{\tL(\varphi) \cap \bB_{\treeZ^r}} =  \standardMeasure{\tL(\varphi^{\ast}) \cap \bB_{\treeZ^r}}.\]
	Moreover, for every tree $\treeZ \in \bB_{\treeZ^r}$ we have that $\treeZ \models \varphi$ if and only if 
	$\treeZ^r \models \varphi^\ast$. We call the formula $\varphi^{\ast}$ the reduction of $\varphi$.
\end{lemma}

\begin{proof}
	If $\varphi$ has a root formula $\varphi_i$ then we take $\varphi^{\ast} \eqdef \exists x. \varphi_i(x) \land \td(x,\varepsilon)<r$.
If $\varphi$ has no root formulae but is satisfiable then we take $\varphi^{\ast} \eqdef \exists x.\rootP(x)$.
Otherwise we take $\varphi^{\ast}\eqdef\bot$.
\end{proof}
\hide{
And in consequence we obtain.
\begin{lemma}
	\label{consequence:basicSentenceMeasure}
	Let $\phi$ be a boolean combination of basic $r$-local sentences and $\treeZ$ be a~complete tree of height $2r+1$. Then $\standardMeasure{\tL(\varphi) \cap \bB_{\treeZ}} = 0 $
	or $\standardMeasure{\tL(\varphi) \cap \bB_{\treeZ}} = \standardMeasure{\bB_{\treeZ}}$.
\end{lemma}

\begin{proof}
The proof is by induction over the structure of the Boolean combination.
For every basic $r$-local sentence $\varphi$ in the Boolean combination, let $\varphi^{\ast}$ be the 
root formula from Corollary~\ref{lemma:uniqueRootFormula}. 
Moreover, let $\phi^*$ be the \emph{reduced} formula, i.e,  the Boolean combination $\phi$ with every basic sentence $\varphi$ replaced with $\exists x.\varphi^\ast(x)$. 

Then the induction statement is as follows.
\begin{itemize}
	\item $\standardMeasure{\tL(\phi) \div \tL(\phi^{\ast})} = 0,$
	\item $
	\standardMeasure{\tL(\phi^\ast) \cap \bB_{\treeZ}} =
	\begin{cases}
	\standardMeasure{\bB_{\treeZ}} & \quad \text{if } \treeZ \models \phi^{\ast}\\
	0 & \quad \text{if } \treeZ \not \models \phi^{\ast}
	\end{cases}
	$
\end{itemize}

If $\varphi$ is a basic $r$-local formula then by Corollary~\ref{lemma:uniqueRootFormula} we are done.
For the induction step let us assume that $\varphi = \varphi_1 \land \varphi_2$,
the remaining connectivities are done similarly.
Then, $\varphi^\ast = \varphi_1^\ast \land \varphi_2^\ast$.
\end{proof}
}

The above result can be extended to Boolean combinations of $r$-local basic formulae by the following
property of measurable sets.

\begin{lemma}
	\label{lemma:booleanAlgebraAndMeasure}
	Let $M$ be a measurable space with measure $\mu$, $S$ be a measurable set and $\{ S_i \}_{i \in I}$ be a family of measurable sets such that for every $i \in I$ either $\mu(S \cap S_i) = 0$ or $\mu(S \cap S_i) = \mu(S)$.
	Then for every set $W$ in the Boolean algebra of sets generated by $\{ S_i \}_{i \in I}$ we have that
	either $\mu(S \cap W) = 0$ or $\mu(S \cap W) = \mu(S)$.
\end{lemma}

\begin{proof}
	The proof goes by a standard inductive argument.
\end{proof}

Hence, by Lemma~\ref{lemma:basicSentenceMeasure} and the above lemma, we obtain the following.

\begin{lemma}
	\label{cor:standardMeasureOfBCOfBasicFormulae}
	Let $\phi$ be a boolean combination of basic $r$-local formulae and  $\treeZ$ be a~complete tree of height $2r+1$.
	Then, $\standardMeasure{ \tL(\phi) \cap  \bB_{\treeZ}} = \standardMeasure{ \tL(\phi^{\ast}) \cap  \bB_{\treeZ}}$, where $\phi^*$ is the reduction of $\phi$, i.e. the Boolean combination $\phi$ with its every basic $r$-local sentence $\varphi$ replaced by its reduction $\varphi^*$, as defined in Lemma~\ref{lemma:uniqueRootFormula}.
	
	Moreover, $\standardMeasure{ \tL(\phi^{\ast}) \cap  \bB_{\treeZ}}=
	\begin{cases}
	\standardMeasure{\bB_{\treeZ}} & \quad \text{if } \treeZ \models  \phi^{\ast};\\
	0 & \quad \text{otherwise.}
	\end{cases}$
\end{lemma}

With above lemmas we can finally prove Theorem~\ref{thm:FOcomputable}.

\begin{proof}[Proof of Theorem~\ref{thm:FOcomputable}]

Let $\varphi$ be a~first order sentence as in the theorem.
We utilise the Gaifman locality theorem~(see Theorem~\ref{thm:gaifman} on page~\pageref{thm:gaifman}) to translate the sentence $\varphi$ into a Boolean combination $\phi$ of basic $r$-local sentences.
Now, let $\phi^{\ast}$ be the reduction of $\phi$, as in Corollary~\ref{cor:standardMeasureOfBCOfBasicFormulae},
and let $S \subseteq \trees{\hasBlank{\Gamma}}$ be the set of all complete trees of height $2r+1$. Then 
$$
\begin{array}{r c l c l }
\standardMeasure{\tL(\phi)}& \eqext{1} & \standardMeasure{\tL(\phi) \cap \big(\bigcup_{\treeZ \in S} \bB_{\treeZ}\big)} & \eqext{2} &  \standardMeasure{\bigcup_{\treeZ \in S} \big(\tL(\phi) \cap  \bB_{\treeZ}\big)}\\ 
&\eqext{3}& \sum_{\treeZ \in S} \standardMeasure{ \tL(\phi) \cap  \bB_{\treeZ}} & \eqext{4}  & 
\sum_{\treeZ \in S} \standardMeasure{ \tL(\phi^{\ast}) \cap  \bB_{\treeZ}}\\ 
& \eqext{5}  & \sum_{\treeZ \in S \land \treeZ \models \phi^\ast} \standardMeasure{\bB_{\treeZ}} & \eqext{6} & |\{\treeZ \in S \mid \treeZ \models \phi^\ast\}|\cdot\frac{1}{2^{2^{r+1}-1}}.
\end{array} 
$$
The first equation follows from the fact that $\{\bB_{\treeZ} \mid \treeZ \in S\}$ is a partition of the space.
The second from operations on sets and the third from a~simple property of measures.
The fourth from the first part of Lemma~\ref{cor:standardMeasureOfBCOfBasicFormulae},
while the fifth follows from the second part of this lemma.
The last equation is a consequence of the fact that $\standardMeasure{\bB_{\treeZ}} = 2^{-|\dom{t}|}$.

Since $\standardMeasure{\tL(\phi)} = \frac{|\{\treeZ \in S \mid \treeZ \models \psi\}|}{2^{2^{r+1}-1}}$, it is enough to count how many complete trees of height $2r+1$ satisfy the reduced combination.
Thus, the complexity upper bound comes from the fact that translating a first-order formula into a Gaifman normal form can produce a three-fold exponential formula in result, see Theorem~\ref{thm:gaifmanCost}, which then can be checked against two-fold exponential number of trees of size that is two-fold exponential in the size of the original formula.
\end{proof}

\begin{algorithm}                      
	\caption{FO measure}          
	\label{alg:FOMeasure}                           
	\begin{algorithmic}                    
		\REQUIRE a $\fo$ formula $\varphi$ 
		\STATE $S \eqdef \text{the set of all of all complete trees of height } 2r+1$ 
		\STATE $\phi \eqdef \text{Gaifman}(\varphi)$ 
		\STATE $\phi \eqdef \text{extractRootFormulae}(\phi)$
		\STATE $S \eqdef \{\treeZ \in S \mid  \treeZ \models \phi\} $

		\RETURN ${|S|}\cdot {2^{-2^{r+1}+1}}$
	\end{algorithmic}
\end{algorithm}

The technique used to prove Theorem~\ref{thm:FOcomputable} cannot be extended to formulae utilising the descendant relation because, as presented in 
Example~\ref{ex:algebraicValueFO}, languages defined by such formulae can have irrational measures.

\begin{proposition}
	\label{ex:algebraicValueFO}
	Let $\Gamma$ be an alphabet. Then there is a language definable in by a~$\fo$ formula over
	the signature $\Gamma \cup \{ s_{\dL}, s_{\dR}, {\ancestor} \}$ with an irrational standard measure.
\end{proposition}

\begin{proof}
Let $\Gamma = \{a,b\}$, we define a~language $L$ in the following way $L \eqdef \{\treeZ \in \trees{\{a,b\}} \mid$ for every path the earliest node labelled $b$
(if exists) is at an even depth$\}$. Now, the measure $\standardMeasure{L}$ is irrational, and 
there is a~language $L'$ definable by a first-order formula over the signature
$\Gamma \cup \{ s_{\dL}, s_{\dR}, {\ancestor} \}$ such that $\standardMeasure{L'} = \standardMeasure{L}$.
We start with computing the measure of $L$, then we will define $L'$.

Observe that the measure $\standardMeasure{L}$ satisfies the following equation.
\[
\standardMeasure{L} = \standardMeasureBig{\{\treeZ \in \treesInfinite{\{a,b\}} \mid \treeZ(\varepsilon) {=} b  \}} +
\standardMeasureBig{\treeZ \in \treesInfinite{\{a,b\}} \mid \treeZ(\varepsilon){=}\treeZ(\dR){=}\treeZ(\dL){=} a } \cdot \standardMeasure{L}^4
\]
After substituting the appropriate values, we obtain the equation
\begin{equation}
\standardMeasure{L} = \frac{1}{2} + \frac{1}{8} \standardMeasure{L}^4
\end{equation}
which by the \emph{rational root theorem} has no rational solutions.

To end the proof, we will describe how to define the language $L'$. The crux of the construction comes 
from the beautiful example by Potthoff, see~\cite[Lemma 5.1.8]{potthoffExample}. 
We will use the following interpretation of the lemma: one can define in first-order logic
over the signature $\{a,b, s_{\dL}, s_{\dR}, {\ancestor} \}$ a language of finite trees over the alphabet $\{a,b\}$ where
every $a$-labelled node has exactly two children end every $b$-labelled node is a leaf on an even depth.

To construct $L'$ we simply utilise the formula defining the language in the Potthoff's example to define $L'$
by substituting the leaf with the first occurrence of the label $b$.
Note that the set $L'$ agrees with $L$ on every tree that 
has a label $b$ on every infinite path from the root.
On the other hand, the truth value of the modified formula on the trees that have an infinite path from the 
root with no $b$-labelled nodes, i.e. on the set $L_{a2}$ from Example~\ref{ex:someMeasuresOfSets}, is of no concern to us.
Indeed, as previously shown, the standard measure of the set $L_{a2}$ is $0$.

To be precise, for every tree $\treeZ \in \treesInfinite{\{a,b\}} \setminus L_{a2}$ we have that
$\treeZ \in L \iff \treeZ \in L'$, where $L_{a2}$ is a language from the Example~\ref{ex:someMeasuresOfSets}.
Therefore, we have that $L \cup L_{a2} = L' \cup L_{a2}$. 
Since $\standardMeasure{L_{a2}} = 0 $, we have that
\[
	\standardMeasure{L} = \standardMeasure{L \cup L_{a2}} = \standardMeasure{L' \cup L_{a2}} = \standardMeasure{L'},
\]
which concludes the proof.
\end{proof}

\section{Conjunctives queries and the standard measure}
\label{sec:CQsMeasure}

Introducing the ancestor/descendant relation to the tree structure causes that every two nodes in the Gaifman graph are in distance at most two from each other.
Thus, for the purpose of having a~relevant definition of the distance in the tree, we retain the child related notion of distance,
i.e. in this section, as before, the notion of the distance is induced by the child relations only.

\paragraph{Conjunctive queries}
A~\emph{concjunctive query} (CQ) over an alphabet $\Gamma$ is a formula of first-order logic, using only conjunction and existential quantification, over unary 
predicates $a(x)$, for $a \in \Gamma$, the root predicate $\rootP(x)$, and binary predicates $ s_{\dL}, s_{\dR}, s, \ancestor$.

An alternative way of looking at conjunctive queries is via graphs and graph homomorphisms. A pattern $\pi$ over $\Gamma$ is 
a relational structure $\pi = \langle V, V_{\rootP}, E_{\dL}, E_{\dR}, E_s, E_{\ancestor}, \lambda_{\pi} \rangle$, where $\parfun{\lambda_{\pi}}{V}{\Gamma}$
is a partial labelling, $V_{\rootP}$ is the set of root vertices, and $G_\pi = \langle V, E_{\dL} \cup E_{\dR} \cup E_s \cup E_{\ancestor} \rangle$ is a finite graph whose edges are split into
left child edges $E_{\dL}$, right child edges $E_{\dR}$, child edges $E_{s}$, and ancestor edges $E_{\ancestor}$. By $|\pi|$ we mean the size of the underlying graph.

We say that a tree $t = \langle \dom{t}, s_{\dL}, s_{\dR}, \ancestor, (a^{t})_{a\in\Gamma} \rangle$ satisfies a pattern $\pi = \langle V, V_{\rootP}, E_{\dL}, E_{\dR}, E_s, E_{\ancestor}, \lambda_{\pi} \rangle$, denoted $t \models \pi$, if there exists a homomorphism $\fun{h}{\pi}{t}$, that is a function $\fun{h}{V}{\dom{t}}$ such that
\begin{enumerate}
    \item \fun{h} {\langle V, E_{\dL}, E_{\dR}, E_s, E_{\ancestor}\rangle} {\langle \dom{t}, s_{\dL}, s_{\dR}, s_{\dL} \cup s_{\dR} , \ancestor \rangle} is a homomorphism of relational structures,
    \item for every $v \in V_{\rootP}$ we have that $h(v) = \varepsilon$,
    \item and for every $v \in \dom{\lambda_{\pi}}$ we have that $\lambda_{\pi}(v) = \lambda_t(h(v))$.
\end{enumerate}

Every pattern can be seen as a conjunctive query and vice versa. Hence, we will use those terms interchangeably.
The class of conjunctive queries is denoted \text{CQ}, the class of formulae that are Boolean combination of conjunctive queries is denoted \text{BCCQ}.

Despite allowing the use of ancestor in conjunctive queries, the measure of the language defined by a~conjunctive query is rational and computable.

\begin{thm}
	\label{thm:CQsAreRational}
	Let $q$ be a conjunctive query over the signature $\Gamma \cup \{\rootP, s_\dL, s_\dR, s, {\ancestor}\}$.
	Then, the measure of the language $\tL(q)$ is rational and computable in exponential space.
\end{thm}
 
To prove the theorem we will use the concept of \emph{firm} sub-patterns, used e.g. in~\cite{murlakAcyclic}.
Intuitively, a \emph{firm pattern} is a conjunctive query that has to be mapped in a small neighbourhood.

A sub-pattern $\pi'$ is \emph{firm} if it is a~sub-pattern of a pattern $\pi$ induced by vertices belonging to a~maximal strongly connected 
component in graph $G_{\pi} = \langle V , E \rangle$ such that $\langle x,y \rangle \in E$ if
either $x s_{\dL} y$, $y s_{\dL} x$, $x s_{\dR} y$, $y s_{\dR} x$, $x s y$, $y s x$, $x {\ancestor} y$, or $\rootP(x)$.
In particular, a pattern is firm if it has a~single strongly connected component.
We say that a sub-pattern is \emph{rooted} if it contains predicate $\rootP$.

\begin{proposition}
	Let $\pi$ be a~firm pattern. Then for every tree $t$ such that $t \models \pi$, for every two vertices $x,y$ in $V$, and for every homomorphism
	$\fun{h}{\pi}{t}$ we have that $d(h(x),h(y)) < |\pi|$.
	Moreover, if $\pi$ is rooted then for every vertex $x$ we have that $d(h(x),\varepsilon) < |\pi|$.
	
\end{proposition}
\begin{proof}

Let us assume otherwise, let $n = |\pi|$. Then, there is a~tree $t$, a~homomorphism $h$, and two vertices $x,y$ such that $t \models \pi$ and $d(h(x),h(y)) \ge n$.
We claim that $x$ and $y$ are not in the same strongly connected component.

Since for some $m$ we have that $d(h(x),h(y)) = m - 1 \ge n$, there is a sequence of distinct nodes $u_1, u_2, \dots u_m$ such that $u_1 = h(x)$,  $u_m = h(y)$ and for every $i$,
$u_i$ and $u_{i+1}$ are in a~child relation.
Moreover, there is a node $u$ such that $u = u_i$ for some $1 \leq i \leq m$, $u \notin h(\pi)$, and
one of the nodes $h(x)$ or $h(y)$ is a descendant of $u$. Without loss of generality, let us say that $u \ancestor h(y)$. Or, more precisely, that $u\dL \weakAncestor h(y)$.

If $x$ and $y$ were in a strongly connected component then there would be a path in the graph $G_\pi$ that connects $y$ to $x$,
i.e. a~sequence of vertices $y_1, y_2, \dots, y_k$, for some $k$, such that $y_1 = y$, $y_k = x$, and for every $i= 1, \dots, k-1$ there is an edge between $y_i$ and $y_{i+1}$ in $G_{\pi}$. In particular, this implies that for every $i$ we have that $h(y_i)$ and $h(y_{i+1})$ are $\weakAncestor$\=/comparable.
Now, there would also exist an index $j \in \{1,\dots, k-1 \}$ such that $h(y_{j+1}) \ancestor u \ancestor h(y_{j})$.
Indeed, if there would be no such index, then all the vertices $y_i$ would satisfy $u\dL \weakAncestor h({y_i})$, as $y_i$ and $y_{i+1}$ are $\weakAncestor$\=/comparable
for every index $i$.
But this impossible because if $u\dL \weakAncestor h(y_i)$ for all $i$, then we would have that $u\dL \weakAncestor h(y_k) = h(x)$.
Now, since $u\dL \weakAncestor h(y)$ and $u\dL \weakAncestor h(x)$, then by the definition of the distance $u$ would not belong to the sequence $u_1, \dots, u_m$.
Which is a contradiction with our assumption.

Therefore, there is an index $j$ such that $h(y_{j+1}) \ancestor u \ancestor h(y_{j})$. Thus, by the definition of $G_{\pi}$ we have that either
$y_{j} s_{\dL} y_{j+1}$, $y_{j+1} s_{\dL} y_{j}$, $y_{j} s_{\dR} y_{j+1}$, $y_{j+1} s_{\dR} y_{j}$, $y_{j} s y_{j+1}$, $y_{j+1} s y_{j}$, $y_{j} {\ancestor} y_{j+1}$, or $\rootP(y_{j})$.
Either of child relations is impossible because the distance between 
$h(y_{j})$ and $h(y_{j+1})$ is at least two. Similarly, both $y_{j} {\ancestor} y_{j+1}$ and $\rootP(y_{j})$ are impossible because
we have that $u \ancestor h(y_{j})$.
Hence, there is no such sequence $y_1, \dots y_k$, thus $x$ and $y$ cannot belong to the same strongly connected component.
This proves the first part of the lemma.

Now, if $\pi$ is rooted then there is a vertex $y$ such that for every homomorphism $h$ we have that $h(y) = \varepsilon$.
Hence, by the first part for all vertices $x\in \pi$ we have that $d(h(x),\varepsilon) = d(h(x),h(y))< n$.
\end{proof}

Let $\pi$ be a pattern.
Consider a graph $\graph^F_\pi = \langle V, E \rangle $ where $V$ is the set of firm sub-patterns of $\pi$ and there is and edge
$\langle v_1, v_2 \rangle \in E \subseteq V \times V$ between two vertices $v_1, v_2 \in V$ if and only if there is an $\ancestor$ labelled edge between some two vertices $w_1 \in v_1,w_2 \in v_2$.
We call this graph the \emph{graph of firm sub-patterns} of the~pattern $\pi$.

\begin{proposition}
	The directed graph $G^F_\pi$ of firm sub-patterns of a pattern $\pi$ is acyclic and has at most one rooted firm sub-pattern.
	We call this sub-pattern \emph{the~root pattern}.
\end{proposition}

\begin{proof}
	By the definition of the firm sub-pattern, every vertex with predicate $\rootP$ ends up in the same maximal strongly connected component.
	The acyclicity follows directly from the fact that firm sub-patterns are the maximal strongly connected components.
\end{proof}

As in the case of root formulae, the root pattern decides of the behaviour of a~satisfiable conjunctive query.
\begin{lemma}
    \label{lemma:rootCQs}
	Let $\Gamma$ be an~alphabet and $q$ be a~conjunctive query over the signature  $\Gamma \cup \{ \rootP, s_\dL, s_\dR, s, {\ancestor} \}$.
	Then, either
	\begin{itemize}
		\item $q$ is not satisfiable and $\standardMeasure{\tL(q)} = 0$,
		\item $q$ is satisfiable,  has no root sub-pattern, and $\standardMeasure{\tL(q)} = 1$,
		\item or $q$ is satisfiable,  has a~root sub-pattern $p$, and $\standardMeasure{\tL(q)} = \standardMeasure{\tL(p)}$.
	\end{itemize}
\end{lemma}

\begin{proof}
	Let $\pi$ be a pattern equivalent to $q$.
	If $q$ is not satisfiable then $\tL(q) = \emptyset$ and $\standardMeasure{\tL(q)} = 0$.
	Let $q$ be satisfiable, i.e. there is a tree $\treeZ^q$ and a homomorphism $\fun{h}{\pi}{\treeZ^q}$.
	Let $\treeZ^r$ be a finite tree such that $h(\pi) \subseteq \nodes{\treeZ^r}$ and let the set $S \subseteq \trees{\nhasBlank{\Gamma}}$ be the set of all trees $\treeZ$ such that for every node $u \in \dS^{|q| +1}$ the tree $\treeZ^r$ is a~sub-tree of $\subtreeAtNode{\treeZ}{u}$. By Lemma~\ref{lemma:basicMeasuresLemma} we have that $\standardMeasure{S} = 1$.

	If $\pi$ has no root firm sub-pattern then $S \subseteq \tL(\pi)$ and we have that 
	\[\standardMeasure{\tL(\pi)} \geq \standardMeasure{S} = 1.\]
	
	On the other hand, if $\pi$ has a~root firm sub-pattern $p$ then for every tree $\treeZ \in S$
	we have that $\treeZ \models \pi$ if and only if $\treeZ \models p$.
	Thus, $\tL(\pi) \cap S = \tL(p) \cap S$ and since $\standardMeasure{S} = 1$ we have that 
	\[
	\standardMeasure{\tL(\pi)} = \standardMeasure{\tL(\pi) \cap S} = \standardMeasure{\tL(p) \cap S} = \standardMeasure{\tL(p)}.
	\]
\end{proof}

In other words, the problem of computing measure of a language defined by a conjunctive query reduces to the below problem of counting 
the models of fixed depth.

\decpro{The models counting problem}{problem:cqCounting}{A conjunctive query $q$ and a natural number $n$.}{Number of complete binary trees of height $n$ that satisfy $q$.}

\begin{proposition}
	The models counting problem can be solved in exponential space.
\end{proposition}
Indeed, all we need is to enumerate all binary trees of height linear in the size of the query.
As an~immediate consequence we infer Theorem~\ref{thm:CQsAreRational}.

For a~lower bound of the problem of computing the standard measure of a conjunctive query, we observe that
deciding whether the measure of a~language defined by a conjunctive query is positive is intractable.
\begin{proposition}
    The positive $\standardMeasure{\text{CQ}}$ problem is \np-complete.
\end{proposition}
\begin{proof}
    Let $q$ be a conjunctive query. Then, either $q$ is not satisfiable and $\standardMeasure{q} = 0$, or
    $q$ is satisfiable and has positive measure.
    That is, $\standardMeasure{q} > 0$ if and only if $q$ is satisfiable.
    Deciding whether a conjunctive query is satisfiable is $\np$-complete, cf. e.g.~\cite{bjorklundCQscontainment}.
\end{proof}

As in the case of first-order formulae, we can lift Theorem~\ref{thm:CQsAreRational} to Boolean combinations of conjunctive queries.
Indeed, by Lemma~\ref{lemma:booleanAlgebraAndMeasure}, the measure of the language defined by a Boolean combination of conjunctive queries is computable.
\begin{corollary}
Let $\varphi$ be a Boolean combination of conjunctive queries. Then, the standard measure $\standardMeasure{\tL(\varphi)}$ can be computed in exponential space.
\end{corollary}

For the lower bound, we observe.
\begin{thm}
    \label{thm:BCCQLowerBound}
    The positive $\standardMeasure{\text{BCCQ}}$ problem is \nexp-complete.
\end{thm}
\begin{proof}[Sketch of the proof]
    First, we prove the upper bound.
    Let $\phi$ be a Boolean combination of conjunctive queries. Let $m$ be the maximum over the sizes of conjunctive queries in $\phi$.
    By Lemma~\ref{lemma:rootCQs} and Lemma~\ref{lemma:booleanAlgebraAndMeasure}, we can translate $\phi$ into a Boolean combination of firm, rooted conjunctive queries $\phi^{\ast}$ such that $\standardMeasure{\tL(\phi)} = \standardMeasure{\tL(\phi^{\ast})}$ and $\phi^{\ast}$ is of polynomial size with respect ot $\phi$.
    This can be done in exponential time and requires verifying whether the patterns in $\phi$ are satisfiable.
    
    Now, since every conjunctive query in $\phi^{\ast}$ is firm and rooted, either there is a~finite tree $t$ of depth $2n$ such that $t \models \phi^\ast$
    or not. If there is no such tree, then $\standardMeasure{\tL(\phi^{\ast})} = 0$. If there is such a tree, then $\setBall{t} \subseteq \tL(\phi^{\ast})$ and by Lemma~\ref{lemma:basicMeasuresLemma} we have that $\standardMeasure{\tL(\phi^{\ast})} > 0$.
    Hence, it is enough to guess the tree $t$ and verify that $t \models \phi^\ast$. Since $t$ is of exponential size in $\phi$ and model checking 
    of a conjunctive query can be done in polynomial time, we infer the upper bound.
    
    For the lower bound, we refer to the proof of Theorem~3 in~Murlak et al.~\cite{murlakAcyclic}, the case of non-recursive schemas.
    The proof can be easily adapted to our needs.
\end{proof}

\section{Conclusions and future work}
We have shown that there exists an~algorithm that, given a~first\=/order sentence $\varphi$ over the signature $\Gamma \cup \{ \rootP, s_{\dR}, s_{\dL}\}$,
computes $\mu^{\ast}(\tL(\varphi))$ in three-fold exponential time.
We also have shown that there exists an~algorithm that, given a~Boolean combination of conjunctive queries $\varphi$ over the signature $\Gamma \cup \{ \rootP, s_{\dR}, s_{\dL}, s, \ancestor \}$, computes the standard measure $\mu^{\ast}(\tL(\varphi))$ in exponential space.
Establishing exact bounds of the problems is an interesting direction of future research.
We provide some lower bounds for the conjunctive queries in the form of the positive measure problem.
We claim, without a~proof, that using the same techniques, we can give similar lower bounds for the 
first-order case.

Note, that the considered measure respects a form of a~$0{-}1$-law. By Lemma~\ref{lemma:basicMeasuresLemma}, if $t$ is a finite tree then with probability $1$ 
it appears as a sub-tree in a~random tree.
It would be interesting to extend the enquiry to measures that do not posses such a property.
Such measures can be expressed, for example, by graphs or by branching boards, cf.~\cite{przybylkoRegularBranchingGamesMFCS}.

Obviously, the most interesting problem is to find an~algorithm that can compute the standard measure of an arbitrary regular language of infinite trees.
While we know that languages with irrational measures exist, we conjecture that for any regular language of trees $L$
the standard measure $\standardMeasure{L}$ is algebraic.

\paragraph{Acknowledgements}
We thank Damian Niwiński and Micha{\l} Skrzypczak for inspiring discussions
and careful reading of a preliminary version of this paper, and the anonymous
referees for helpful comments motivating us to improve the presentation of the
paper.

\bibliographystyle{eptcs}
\bibliography{biblio,mskrzypczak}

\end{document}

%% file: macros.tex
\usepackage{microtype}

\usepackage[utf8]{inputenc}
\usepackage{amsfonts}
\usepackage{amsmath, amssymb}
\usepackage{graphicx}
\usepackage{wrapfig}
\usepackage{tikz}
\usetikzlibrary{shapes,decorations,calc}
\usepackage{todonotes}
\usepackage{xspace}
\usepackage{subcaption}
\usepackage{mathdots}

\usepackage{fdsymbol}

\usepackage[shortcuts]{extdash} 

\usepackage{listings}
\usepackage[shortcuts]{extdash}

\usepackage{color}
\usepackage{algorithm}
\usepackage{algorithmic}

\newcommand{\hide}[1]{}

\newcommand{\treeZ}{t}

\newcommand{\eqdef}{\stackrel{\text{def}}=}
\newcommand{\eqext}[1]{\stackrel{#1}=}

\newcommand{\comment}[1]{}

\newcommand{\dL}{\mathtt{\scriptstyle L}}
\newcommand{\dR}{\mathtt{\scriptstyle R}}

\newcommand{\dS}{\{\dL, \dR\}\xspace}

\newcommand{\setD}{\dS}
\newcommand{\setPositions}{\{ \dL, \dR \}^{\ast}}
\newcommand{\setNodes}{\dS^{\ast}}

\newcommand{\graph}{\mathrm{G}}

\renewcommand{\AA}{\mbox{$\mathcal{A}$}}

\newcommand{\hasBlank}[1] {#1}
\newcommand{\nhasBlank}[1] {#1}

\renewcommand{\AA}{\mathcal{A}}
\newcommand{\CC}{\mathcal{C}}

\newcommand{\TT} {\mathcal{T}}

\newcommand{\bB}{\mathbb{B}}

\newcommand{\tL}{\mathrm{L}}
\newcommand{\td}{\text{d}}

\newcommand{\trees}[1] {\mathcal{T}_{#1}\xspace}
\newcommand{\treesInfinite}[1] {\mathcal{T}_{#1}^{\omega}\xspace}
\newcommand{\setBall}[1]{\mathbb{B}_{#1}}
\newcommand{\setBallInNode}[2]{\mathbb{B}_{#1,#2}}

\newcommand{\dom}[1] {\textit{Dom}(#1)}

\newcommand{\nodes}[1] {\dom{#1}}

\newcommand{\fun}[3]{\ensuremath{#1\colon #2 \to #3}}
\newcommand{\parfun}[3]{\ensuremath{#1\colon #2 \rightharpoonup #3}}

\newcommand{\rootP}{\varepsilon}

\newcommand{\ancestor}{\sqsubsetneq}

\newcommand{\weakAncestor}{\sqsubseteq}

\newcommand{\subtreeSymbol}{.}
\newcommand{\subtree}[2]{#1{\subtreeSymbol}#2}
\newcommand{\isPrefix}[2]{#1{\sqsubseteq}#2}
\newcommand{\prefix}{\sqsubseteq}

\newcommand{\subtreeAtNode}[2]{{#1}{\subtreeSymbol}#2}

\newcommand{\fo}{\mathrm{FO}\xspace}

\newcommand{\decpro}[4]{
\begin{problem}[#1]
	\label{#2}
	\[\begin{tabular}{r p{12cm}}
		Input:  & #3 \\
		Output: & #4 \\
	\end{tabular}
	\]
\end{problem}
}

\newcommand{\nexp}{\mbox{\textit{NEXP}}}

\newcommand{\np}{\textit{NP}}

\newcommand{\standardMeasure}[1]{\mu^{\ast}(#1)}
\newcommand{\standardMeasureBig}[1]{\mu^{\ast}\big(#1\big)}

\tikzstyle{eve}       = [scale=0.8, inner sep=0, shape=diamond,draw=black,minimum size=10mm]
\tikzstyle{adam}      = [scale=0.8, inner sep=0, shape=rectangle,draw=black,minimum size=8mm]
\tikzstyle{branching} = [scale=0.8, inner sep=0, regular polygon, regular polygon sides=3 ,draw=black,minimum size=10mm]
\tikzstyle{generic} = [scale=0.8, inner sep=0, regular polygon, regular polygon sides=8 ,draw=black,minimum size=10mm]
\tikzstyle{nature}    = [scale=0.8, circle, draw=black, minimum size=9mm]

\tikzstyle{successor} = [->, draw=black!80]